\documentclass[letterpaper,11pt,english]{article}

\usepackage[letterpaper, portrait, margin=1in]{geometry}

\usepackage{amsmath,amsthm}\allowdisplaybreaks
\usepackage{amssymb}
\usepackage{mathtools}
\usepackage{csquotes}
\usepackage{graphicx}
\usepackage{grffile}
\usepackage{subfig}
\usepackage[hidelinks]{hyperref}

\usepackage{soul}
\usepackage{xcolor}
\usepackage{xspace}
\usepackage{paralist}
 
\usepackage{algorithm}
\usepackage[noend]{algorithmic}
\usepackage{textcomp}
\usepackage{nicefrac}

 \usepackage[disable]{todonotes} % hide todo notes
\usepackage[subtle]{savetrees}

\theoremstyle{plain}
\newtheorem{theorem}{Theorem}
\newtheorem{lemma}[theorem]{Lemma}

\newtheorem{corollary}[theorem]{Corollary}

\newcommand{\aboveAndBelowSkip}{4pt}

\newcommand{\calO}{\mathcal{O}}
\newcommand\tab[1][0.8cm]{\hspace*{#1}}

\newenvironment{tightalign*}
 {
 \setlength{\abovedisplayshortskip}{\aboveAndBelowSkip}\setlength{\abovedisplayskip}{\aboveAndBelowSkip}
 \setlength{\belowdisplayshortskip}{\aboveAndBelowSkip}\setlength{\belowdisplayskip}{\aboveAndBelowSkip}
 \csname align*\endcsname}
 {\endalign}

\usepackage{authblk}

\title{The Mobile Server Problem
\thanks{
A preliminary version of this paper appeared in SPAA'17~\cite{Feldkord17}.
\newline
This work is partially supported by the German Research Foundation (DFG) within the Collaborative Research Center ``On-The-Fly Computing'' (SFB 901).}}
\author[]{Bj\"orn Feldkord}
\author[]{Friedhelm Meyer auf der Heide}
\affil[]{Heinz Nixdorf Institute and Department of Computer Science\\
	Paderborn University, F\"urstenallee 11, 33102 Paderborn, Germany
}
\affil[]{ \{bjoern.feldkord, fmadh\}@upb.de}
\date{}

\begin{document}
\hypersetup{pageanchor=false}
\maketitle

\begin{abstract}
  We introduce the mobile server problem, inspired by current trends to move computational tasks from cloud structures to multiple devices close to the end user.
  An example for this are embedded systems in autonomous cars that communicate in order to coordinate their actions.
	
	Our model is a variant of the classical Page Migration Problem.
	More formally, we consider a mobile server holding a data page.
	The server can move in the Euclidean space (of arbitrary dimension).
	In every round, requests for data items from the page pop up at arbitrary points in the space.
	The requests are served, each at a cost of the distance from the requesting point and the server, and the mobile server may move, at a cost $D$ times the distance traveled
	for some constant $D$.
	We assume a maximum distance $m$ the server is allowed to move per round.
	
	We show that no online algorithm can achieve a competitive ratio independent of the length of the input sequence in this setting.
	Hence we augment the maximum movement distance of the online algorithms to $(1+\delta)$ times the maximum distance of the offline solution.
	We provide a deterministic algorithm which is simple to describe and works for multiple variants of our problem.
	The algorithm achieves almost tight
	competitive ratios independent of the length of the input sequence.
	
	Our Algorithm also achieves a constant competitive ratio without resource augmentation in a variant
	where the distance between two consecutive requests is restricted to a constant smaller than the limit for the server.
\end{abstract}

\thispagestyle{empty}
\clearpage
\setcounter{page}{1}
\pagebreak
\hypersetup{pageanchor=true}

%%%%%%%%%%%%%%%%%%%%%%%%%%%%%%%%%%%%%%%%%%%%%%%%%%%%%%%%%%%%%%%%%%%%%%%%%%%%%%%%%%
%%%%%%%%%%%%%%%%%%%%%%%%%%%%%%%%%%%%%%%%%%%%%%%%%%%%%%%%%%%%%%%%%%%%%%%%%%%%%%%%%%

\section{Introduction}\label{sec:intro}
Motivated by their large consumption of computational resources, a growing number of applications
were shifted from a single machine at some end user
to large computing centers.
These centers may have networks of processors working on a common memory to execute some
common computational task.
This development has motivated lots of research problems regarding a static network of
machines with some common resources such as memory and bandwidth of common
communication channel (see~\cite{Jansen16, Kling17} for recent examples in the area of scheduling).

Recently, another trend can be observed where the computation is shifting back to the user.
This is commonly referred to as "edge computing" and involves a dynamic network with lots of
(mobile) devices which are located close to the user.
The machines may even be embedded systems such as in autonomous cars which need to share
data in order to coordinate.
As a result, the structure of the communication network might be much more complex in the sense
that devices may join and leave the network and data may be shifted to nodes which are altering
their physical location throughout the computation.
For examples, see~\cite{Ahmed16,Davis04}.

The Page Migration Problem is a simple model for approaching the problem with shared memory in
a cloud computing scenario.
In the classical version of this problem we consider a memory page which is shared
by multiple processors which are connected in a network. Only one processor
can hold the page at a time.
Processors request data from the page which incurs cost proportional to the distance
within the network.
In order to reduce these costs, the page may be moved to another processor; this however
incurs cost proportional to the distance in the network times the size of the page.

In this paper, we present with the Mobile Server Problem a simple model which captures both the idea 
of Page Migration, namely modeling costs for accessing data items from an indivisible data page,
and allowing mobility of data servers in order to improve the overall performance of the system.

Abstracting from specific network topologies, we replace the network graph with the Euclidean space,
such that, at every step, an arbitrary (finite) number of requests for data items can appear anywhere in the space.
In order to improve the overall performance, the mobile server holding the memory page can move to any point in the plane.
For the purpose of bounding the time needed for a round, we limit the allowed distance the server can move in a step.
The cost for a step is composed of the sum of the distances between the mobile server and the positions of the clients issuing the requests,
plus the distance the server moves, weighted with a cost value $D$.

The Page Migration Problem is naturally often considered as an online problem
as data requests may not be known at the start of a larger computation.
This is especially true in our scenario where even the participating
devices may not be known in advance.
Hence, we also consider the Mobile Server Problem as an online problem.

We mainly consider the variant where, in a step, a server may move knowing the position of the clients, but answer their requests afterwards.
We also discuss other variants.

\subsection{Related Work}
The Page Migration Problem was first considered as an online problem by Black and Sleator~\cite{Black89}
who gave 3-competitive algorithms for arbitrary trees and the complete graph.
It was also shown that 3 is the best possible competitive ratio even if the given network just
consists of two processors. This lower bound also holds for randomized algorithms against adaptive
adversaries.
The first deterministic algorithm for general graphs was given by Westbrook in 1994~\cite{Westbrook94}
which was called \textit{Move-To-Min} as the strategy is to move to the optimal point in regards to
the last $D$ requests.
The competitive ratio of this algorithm is 7.
For the randomized solutions, there is a simple 3-competitive algorithm called the \textit{Coin-Flip Algorithm}
which is simple to describe and also works against adaptive online adversaries.
A more involved solution gives a 2.62-competitive algorithm against oblivious adversaries.
These algorithms can be found in a paper by Westbrook~\cite{Westbrook94}.
An overview over these results including better results for the deterministic case can be found in a
survey by Yair Bartal~\cite{Bartal98}.

The Page Migration Problem belongs to the class of \textit{relaxed task systems}. This was used to
derive a deterministic solution for the problem with multiple copies of the page. The adaption within the framework
is made from an online algorithm for the $k$-Server Problem~\cite{Bartal01}.
The $k$-Server Problem can be formulated in terms of the Page Migration Problem with the restriction that
requests which appear in the network must be satisfied by moving one of the $k$ identical copies of the page
to the location of the request. The competitive ratio of the problem is shown to be $\Omega(k)$.
An overview of the most important results can be found in a survey by Elias Koutsoupias~\cite{Koutsoupias09}.
A recent paper by Böckenhauer et.~al. also studied the advice complexity of the $k$-Server Problem where they
design a $\frac{1}{1-2\sin(\frac{\pi}{2^b})}$-competitive algorithm reading $b$ bits per time step~\cite{Boeckenhauer17}.

In contrast to the $k$-Server model, Albers and Koga studied a variant with multiple different pages (only one copy per page)
but where the processors only can hold one page at a time. The competitive ratio for this model can be bounded by a constant~\cite{Albers95}.

The Page Migration Problem has already been studied within the Manhattan and Euclidean plane, but without restricting the moving
distance of the page in each time step~\cite{Chrobak97, Khorramian16}.
The Manhattan and Euclidean plane were also considered for the 3-Server Problem for which (almost) optimal online algorithms
were given~\cite{Bein02}.

All the mentioned variants so far assumed a static network where the distances between the nodes do not change over time.
Bienkowski et.~al. considered a scenario where distances between nodes could change over time. The change in distance was
done either by an adversary or determined by a stochastic process. In this scenario the competitive ratio depends
both on the size of the page and the number of processors~\cite{Bienkowski09}.

Since the competitive ratio for our problem does depend on the length of the input sequence, we analyze it using the concept
of \textit{resource augmentation}. This technique of giving the online algorithm slightly more power in some sense to
get a bounded competitive ratio was first used by Kalyanasundaram and Pruhs for scheduling problems where the online algorithms
were given slightly faster processors than the offline solution to improve from an unbounded ratio to a competitive ratio
only dependent on the augmentation factor~\cite{Pruhs95}.

\subsection{Our Contribution}
We present the Mobile Server Problem introduced above (and formally described in Section~\ref{secModel}).
Our first results are lower bounds for the competitive ratio:
We prove that no online algorithm can achieve a competitive ratio independent of the number of rounds for this problem.

We therefore consider the problem in a setting where the maximum distance a server may move is augmented by a factor of $(1+\delta)$
for the online algorithm.
We show that a simple algorithm is sufficient to achieve a competitive ratio independent of the length of the input sequence
for several variants of the problem.
In particular, we describe a deterministic algorithm which is $\mathcal{O}(\frac{1}{\delta^{\nicefrac{3}{2}}})$-competitive
in the Euclidean Plane when the number of requests per round is fixed.
We also briefly sketch how to modify our analysis to work for a varying number of requests
per round and for a scenario where requests must be answered before the page can be moved.

To complement these results we show lower bounds for all of the variants which also hold
for randomized algorithms against oblivious adversaries. Overall we get tight bounds for
the 1-dimensional Euclidean space and almost tight bounds, up to a factor of $\frac{1}{\sqrt{\delta}}$,
for arbitrary dimensions.

We also consider a variant where the requests are posed by agents in the network who also move at a limited speed.
We show that if our server can move at least as fast as these agents, our algorithm from before has a constant
competitive ratio, even without the use of resource augmentation.

%%%%%%%%%%%%%%%%%%%%%%%%%%%%%%%%%%%%%%%%%%%%%%%%%%%%%%%%%%%%%%%%%%%%%%%%%%%%%%%%%%
%%%%%%%%%%%%%%%%%%%%%%%%%%%%%%%%%%%%%%%%%%%%%%%%%%%%%%%%%%%%%%%%%%%%%%%%%%%%%%%%%%
\section{Our Model}\label{secModel}
The model of the Mobile Server Problem adapts some notation from the Page Migration Problem
to allow for a better understanding and an easy comparison of the results.

We consider a mobile server holding a memory page, located at a point $P_t$ in the Euclidean Plane at time $t$.
Time is discrete and divided into steps.
We refer to the length of an input sequence as the number of time steps denoted by $T$.

In each time step, clients can request data items from the page.
Let $r_t$ be the number of clients requesting data items in step $t$.
These are represented by points $v_{t,1},\ldots,v_{t,r_t}$ in the plane.
The server can move in every time step by a distance of at most $m$, i.e. $d(P_t,P_{t+1})\leq m$.

The cost for answering a request issued at position $v$ when the server is located at position $P$ is $d(P,v)$.
Moving the server from position $P_t$ to $P_{t+1}$ induces cost $D\cdot d(P_t,P_{t+1})$ for some constant $D\geq 1$.
An algorithm operating on an input sequence for the problem can decide in each step
where to move the page under the given restriction of the distance.

The total costs of an algorithm $Alg$ on a given input sequence are defined as follows:
$$C_{Alg}=\sum_{t=1}^{T}\left(D\cdot d(P_t,P_{t+1})+\sum_{i=1}^{r_t}d(P_{t+1},v_{t,i})\right)$$
Note that in this definition, the page may be moved upon knowing the current requests.
The requests are however served after the page has been moved, hence their answering costs are proportional to
the distance to $P_{t+1}$.

In the following sections we will also discuss a variant we refer to as the Answer-First Variant, where
the requests have to be served before moving the server.
The costs of the algorithm in this case are
$$C_{Alg}=\sum_{t=1}^{T}\left(\sum_{i=1}^{r_t}d(P_{t},v_{t,i}) + D\cdot d(P_t,P_{t+1})\right).$$
We will observe that this small modification of the problem has a huge impact on the best possible
competitive ratio.

%%%%%%%%%%%%%%%%%%%%%%%%%%%%%%%%%%%%%%%%%%%%%%%%%%%%%%%%%%%%%%%%%%%%%%%%%%%%%%%%%%
%%%%%%%%%%%%%%%%%%%%%%%%%%%%%%%%%%%%%%%%%%%%%%%%%%%%%%%%%%%%%%%%%%%%%%%%%%%%%%%%%%
\section{Lower Bounds}

In this section we show how the different parameters of our model influence the quality of the best
possible approximation by an online algorithm.
In this paper we only present a deterministic online algorithm for the Mobile Server Problem, however we formulate the lower bounds
for the expected competitive ratio of randomized online algorithms against oblivious adversaries.
These lower bounds carry over to the deterministic case and to randomized online algorithms against stronger (adaptive) adversaries.
It should also be noted that these lower bounds hold in the Euclidean space for an arbitrary dimension.

For the lower bounds we use Yao's Min-Max Principle~\cite{Yao77} which allows us to construct the lower bounds by generating a randomized input sequence
and examine the expected competitive ratio of a deterministic online algorithm.
We do this in all of the following three theorems and hence omit a reference to this technique in the individual proofs.

First we show that without use of resource augmentation, there does not exist an online algorithm with a competitive ratio independent of $T$.

\begin{theorem}
\label{ThInfinite}
Every randomized online algorithm for the Mobile Server Problem has a competitive ratio of $\Omega(\nicefrac{\sqrt{T}}{D})$ against an oblivious adversary,
even if there is only one request per time step.
\end{theorem}

\begin{proof}
We consider a sequence of $x$ time steps with one request each on the starting position of the server.
Consider two opposite directions from the starting position which we will refer to as left and right.
The adversary decides with probability $\frac{1}{2}$ at the beginning of the first step to either move its server
a distance $m$ to the left or to the right for the first $x$ time steps.
The cost for the adversary is at most $xDm+m\cdot\sum_{i=1}^{x}i\leq xDm+mx^2$ for these first $x$ steps.

After these $x$ steps, with probability $\frac{1}{2}$ the server of the online algorithm has a distance of
at least $xm$ to the server of the adversary.

For the remaining $T-x$ steps of the sequence the adversary issues requests on the position of its server and moves it a distance of $m$ towards
the same direction it already did during the first $x$ steps.
The costs for the adversary are $(T-x)Dm$ while the costs for the online algorithm are at least $(T-x)\cdot xm$ with probability $\frac{1}{2}$.
By choosing $x=\sqrt{T}$ the expected competitive ratio is $\Omega(\frac{\sqrt{T}}{D})$.
\end{proof}

In order to have a chance to achieve a competitive ratio independent from $T$,
we augment the maximum distance by which the server may move in every time step for the respective online algorithm.
We consider online algorithms which, in every round, can move their server by a distance of $(1+\delta)m$ for some $\delta\in(0,1]$.
We do not consider bigger values for $\delta$ in this paper since $1$ is a natural lower bound for all online problems.
Hence asymptotically, for an online algorithm it is sufficient to utilize a distance of at most $2m$ to achieve an
optimal competitive ratio of $\mathcal{O}(1)$.

\begin{theorem}
Let $R_{min}$ and $R_{max}$ be the minimum and maximum number of requests per time step.
Every randomized online algorithm using an augmented maximum moving distance of at most $(1+\delta)m$
has an expected competitive ratio of $\Omega(\frac{1}{\delta}\cdot\frac{R_{max}}{R_{min}})$ against an oblivious adversary.
\end{theorem}

\begin{proof}
We use the same sequence as in the previous theorem to separate the servers of the adversary and the online algorithm:

For $x$ time steps, there are $R_{min}$ requests in every step on the starting position of the server.
The adversary moves its server a distance $m$ for $x$ steps to the left or right with probability $\frac{1}{2}$ each,
such that the distance between the two servers is at least $xm$ with probability at least $\frac{1}{2}$ after these steps.
The costs for the adversary are at most $Dxm+R_{min}mx^2$.

The adversary now issues $R_{max}$ requests at the position of its server and moves it by a distance $m$ in the same direction as in the previous steps.
This is done for exactly as many time steps as it would take the server of the online algorithm to "catch up" with the server of the adversary when it
it is at least a distance $xm$ away from the adversary's server
and moves towards it with the maximum distance in each round.
The necessary number of steps for that are $\frac{x}{\delta}$, since the distance between the two servers decreases by at most $\delta m$ in every round.

The costs the online algorithm has to pay, in case the distance between the two servers is at least $xm$ at the beginning of this phase,
for serving the requests are minimized when the algorithm moves the server with a maximum distance towards
the position of the adversary's server in every time step.
The costs for the online algorithm are therefore at least
$$\begin{array}{rcl}
  R_{max}\cdot\sum_{i=1}^{\frac{x}{\delta}}\left(xm-i\delta m\right) &=& R_{max}\left(\frac{mx^2}{2\delta}-\frac{mx}{2}\right) \\
    &\geq& \frac{1}{4}\frac{R_{max}mx^2}{\delta}
\end{array}$$
with probability at least $\frac{1}{2}$ by choosing $x\geq2\delta$.
The adversary pays $\frac{x}{\delta}Dm$ in this phase.

By choosing $x$ sufficiently large, the total costs of the adversary sum up to at most $3R_{min}mx^2$ over both described phases.

The adversary can now repeat the two phases in a circular way arbitrarily often. The costs can be analyzed as above since the one probabilistic choice
the adversary does is made independently of the behavior of the online algorithm and its own former behavior.
The resulting expected competitive ratio is $\Omega(\frac{1}{\delta}\cdot\frac{R_{max}}{R_{min}})$.
\end{proof}

We observe that as a special case, when $R_{max}=R_{min}$ the lower bound of the competitive ratio becomes independent of the number of
requests in each round.
In the following section we show that this is indeed possible to achieve in this scenario.

We finally address our decision to allow the algorithms to move the page before answering the requests.
Consider the scenario in which an algorithm has to answer the requests before moving the server.
It can be shown that the competitive ratio of such algorithms depend on the number of requests in each time step even if it is fixed throughout the whole sequence.

\begin{theorem}
If in every time step, requests must be answered before moving the server, every randomized algorithm has an expected competitive ratio of $\Omega(\nicefrac{r}{D})$ against an oblivious adversary
if the number of requests in each time step is fixed to a constant $r$.
\end{theorem}

\begin{proof}
Consider the following two time steps:
In the first step, $r$ requests are issued at the common position of the servers.
The adversary can now move the server to a position such that with probability at least $\frac{1}{2}$, the distance between the two servers is at least $m$.
This can be done by throwing a fair coin and the moving to the left or right as in the previous theorems.

In the second step, $r$ requests are issued at the position of the adversary's page.
The two steps may be repeated in a cyclic manner since the random choice of the adversary does not depend on any former time steps.

The costs for the online algorithm for one repetition of these two steps are at least $rm$ with probability at least $\frac{1}{2}$ while the costs of the adversary are at most $Dm$.
\end{proof}

%%%%%%%%%%%%%%%%%%%%%%%%%%%%%%%%%%%%%%%%%%%%%%%%%%%%%%%%%%%%%%%%%%%%%%%%%%%%%%%%%%
%%%%%%%%%%%%%%%%%%%%%%%%%%%%%%%%%%%%%%%%%%%%%%%%%%%%%%%%%%%%%%%%%%%%%%%%%%%%%%%%%%
\section{The Move-to-Center Algorithm}
In this section we provide a simple algorithm which achieves an optimal competitive
ratio on line segments and a near optimal competitive ratio for the Euclidean plane.
\smallskip

The algorithm \textit{Move-to-Center (MtC)} works as follows:

Assume the algorithm has its server located at a point $P_{Alg}$ and receives requests
$v_1,\ldots,v_r$. Let $c$ be the point minimizing $\sum_{i=1}^{r}d(c,v_i)$.
If $c$ is not unique, pick the one minimizing $d(P_{Alg},c)$.
MtC moves the server towards $c$ for a distance of $\min\{1,\frac{r}{D}\}\cdot d(P_{Alg},c)$
if this distance is less than $(1+\delta)m$.
Otherwise it moves the server a distance of $(1+\delta)m$ towards $c$.
\smallskip

The goal of the following analysis is to proof the following theorem:
\begin{theorem}
\label{ThMain}
Let $R_{min}$ and $R_{max}$ be the minimum and maximum number of requests per time step.
Algorithm \textit{MtC} is $\mathcal{O}(\frac{1}{\delta}\cdot\frac{R_{max}}{R_{min}})$-competitive on the infinite line and
$\mathcal{O}(\frac{1}{\delta^{\nicefrac{3}{2}}}\cdot\frac{R_{max}}{R_{min}})$-competitive in the Euclidean plane
using an augmented maximum moving distance of $(1+\delta)m$ for some $\delta\in(0,1]$.
\end{theorem}

We first analyze the algorithm for the case that some fixed number $r$ appear per time, i.e. $R_{max}=R_{min}=:r$.
We then describe how the result of this analysis implies the results for the
general case and for the Answer-First Variant in Section~\ref{sec:ext}.
Note that in our proof we do not optimize the constants and instead focus on readability.

We first show that it is sufficient to analyze a simplified version of the problem in which the $r$ requests occur exactly on the point $c$
picked by MtC in every round.

\begin{lemma}
If MtC is $\alpha$-competitive in an instance where all requests in a time step $t$ occur on one point $c_t$, MtC is $4\alpha+1$-competitive
in an instance where $c_t$ is the closest center of the requests in the respective time step.
\end{lemma}

\begin{proof}
Fix a time step and let $v_1,\ldots,v_r$ be the requests in the original instance where $c$ is the closest center to MtC and $c'$ the closest center to
the optimal solution $OPT$ when serving the requests.
Let $C_{Opt}$ and $C_{Alg}$ be the costs of the optimal solution and the MtC algorithm in the original instance and $C_{Opt}'$ and $C_{Alg}'$ be the respective costs
in the simplified instance for serving the requests.
Let $o$ be the position of the optimal server and $a$ be the position of the algorithm's server.
It holds
$C_{Opt}'=r\cdot d(o,c)\leq r\cdot d(o,c')+r\cdot d(c',c)\leq \sum_{i=1}^{r}d(o,v_i)+3\sum_{i=1}^{r}d(c',v_i)\leq 4C_{Opt}$.
For the algorithm, we have
$C_{Alg}=\sum_{i=1}^{r}d(a,v_i)\leq r\cdot d(a,c)+\sum_{i=1}^{r}d(c,v_i)\leq C_{Alg}'+C_{Opt}$.
Therefore
$\frac{C_{Alg}}{C_{Opt}}\leq 4\frac{C_{Alg}'}{C_{Opt}'}+1$.
\end{proof}

For the analysis of the algorithm, we use a potential function argument.
Therefore we fix an arbitrary time step in the input sequence
and introduce some notation for this step.

By $P_{Alg}$ and $P_{Alg}'$, we denote the position of the algorithm's server before and
after moving it. $P_{Opt}$ and $P_{Opt}'$ will be used for the optimal server positions
before and after moving respectively.

For better readability, we define the following abbreviations which we use for the distances and sometimes also for the lines as geometric objects:
$a_1:=d(P_{Alg},P_{Alg}')$, $a_2:=d(P_{Alg}',c)$, $s_1:=d(P_{Opt},P_{Opt'})$, $s_2:=d(P_{Opt}',c)$, $p:=d(P_{Opt},P_{Alg})$,
$h:=d(P_{Opt}',P_{Alg})$, $h':=d(P_{Opt},P_{Alg}')$ and $q:=d(P_{Opt}',P_{Alg}')$.
An illustration can be found in Figure~\ref{fig:geo1}.

\begin{figure}
\centering
\includegraphics[width=0.7\textwidth]{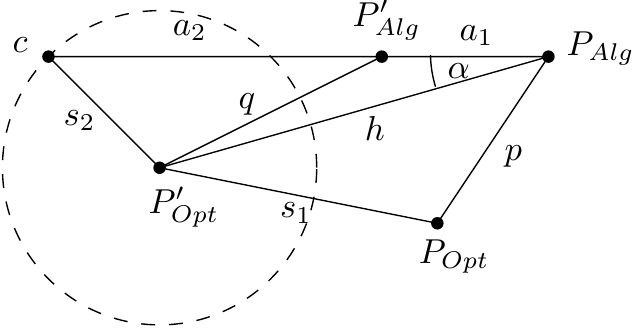}
\caption{Illustration of relevant points and distances for estimating the potential difference.
$h'=d(P_{Opt},P_{Alg}')$ is omitted for better overview.}
\label{fig:geo1}
\end{figure}

Using this notation, the costs of the online algorithm are
$$C_{Alg}=Da_1+ra_2$$
and the optimal costs are
$$C_{Opt}=Ds_1+rs_2.$$

We start with an estimation of the difference $h-q$ which is essential for utilizing the respective potential
functions in the upcoming parts of the analysis.

\begin{lemma}
\label{leGeo}
If $s_2\leq\frac{\sqrt\delta}{1+\frac{1}{2}\delta}a_2$, then $h-q\geq\frac{1+\frac{1}{2}\delta}{1+\delta}a_1$.
\end{lemma}

\begin{proof}
We want to get a lower bound for $h-q$ given fixed values for $h,s_2$ and $a_1$.
$q$ is maximized by choosing the angle $\alpha$ between $a_1$ and $h$ as large as possible.
This can be done by setting the angle between $s_2$ and $a_2$ to 90 degrees as shown in Figure~\ref{fig:geo2}.

\begin{figure}
\centering
\includegraphics[width=0.7\textwidth]{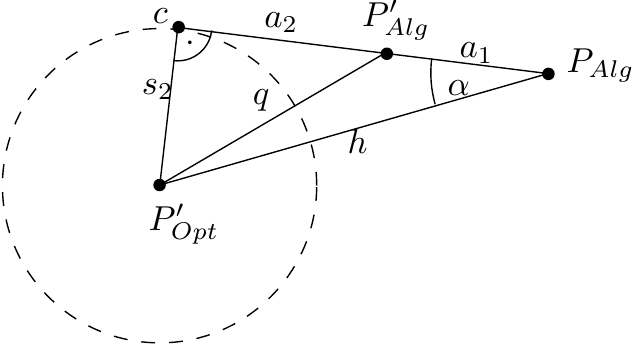}
\caption{Setting where $q$ is maximized given $h,s_2$ and $a_1$.}
\label{fig:geo2}
\end{figure}

Making use of the 90 degree angle we get $h^2=s_2^2+(a_1+a_2)^2$ and $q^2=s_2^2+a_2^2$. Let $s_2=\epsilon\cdot a_2$.
Then
$$\begin{array}{rcl}
  h-q &=& \sqrt{(\epsilon a_2)^2+(a_1+a_2)^2} - \sqrt{(\epsilon a_2)^2+a_2^2} \\
	  &=& \sqrt{(1+\epsilon^2)a_2^2+2a_1a_2+a_1^2} - \sqrt{1+\epsilon^2}a_2 \\
		&\geq& \sqrt{(1+\epsilon^2)a_2^2+2\sqrt{1+\epsilon^2}a_2\frac{a_1}{\sqrt{1+\epsilon^2}}+\frac{a_1^2}{1+\epsilon^2}} 
		   - \sqrt{1+\epsilon^2}a_2 \\
		&=&\sqrt{\left(\sqrt{1+\epsilon^2}a_2+\frac{a_1}{\sqrt{1+\epsilon^2}}\right)^2} - \sqrt{1+\epsilon^2}a_2 \\
		&=& \frac{a_1}{\sqrt{1+\epsilon^2}}.
\end{array}$$

It therefore holds
$$\epsilon\leq\sqrt{\frac{(1+\delta)^2}{(1+\frac{1}{2}\delta)^2}-1}\Rightarrow h-q\geq\frac{1+\frac{1}{2}\delta}{1+\delta}a_1.$$

Since
$$
  \sqrt{\frac{(1+\delta)^2}{(1+\frac{1}{2}\delta)^2}-1} = \frac{\sqrt{\delta+\frac{3}{4}\delta^2}}{1+\frac{1}{2}\delta}
	  \geq \frac{\sqrt{\delta}}{1+\frac{1}{2}\delta}
$$
this implies the lemma.
\end{proof}

\subsection{\texorpdfstring{Analysis for $r>D$}{Analysis for r>D}}
\label{sec:ana1}
In the following, we analyze the algorithm for the case $r>D$ for both the line segment and the Euclidean plane.
Most arguments apply to both cases, if there is a difference between the two it is stated explicitly in the respective case.

For server locations $P_{Alg}$ and $P_{Opt}$ of the online algorithm and the optimal solution respectively,
the potential $\phi$ is defined as
$$\phi(P_{Opt},P_{Alg}):=\left\{\begin{array}{l}
8\frac{r}{\delta m}\cdot d(P_{Opt},P_{Alg})^2,\\
  \tab\text{for }d(P_{Opt},P_{Alg})>\delta\frac{Dm}{4r}\\
2D\cdot d(P_{Opt},P_{Alg}),\\
\tab\text{otherwise}
\end{array}\right..$$

For each of the listed cases we bound the potential difference $\Delta\phi=\phi(P_{Opt}',P_{Alg}')-\phi(P_{Opt},P_{Alg})$ and the costs of the online algorithm
as expressions in $C_{Opt}$.

\textbf{1}. $p\leq\delta\frac{Dm}{4r}\leq\delta m$:
The algorithm has the possibility to move its server to the same position as the optimal solution (or closer to $c$).
Hence $a_2\leq s_2$.
If $q\leq\delta\frac{Dm}{4r}$ we have
$$\begin{array}{rcl}
  \Delta\phi &=& 2D(q-p) \\
	  &\leq& 2r(s_2+a_2)-2Dp \\
		&\leq& 4C_{Opt} - 2Dp
\end{array}$$
and
$$\begin{array}{rcl}
  C_{Alg} &=& Da_1+ra_2 \\
	  &\leq& D(p+s_1+q)+ra_2 \\
		&\leq& D(p+s_1)+r(s_2+2a_2) \\
		&\leq& 3C_{Opt} + Dp.
\end{array}$$

Otherwise it holds $q>\delta\frac{Dm}{4r}\geq p$ and we get
$$\begin{array}{rcl}
  \Delta\phi &=& 8\frac{r}{\delta m}\cdot q^2-2D\cdot p \\
    &\leq& 8\frac{r}{\delta m}(q^2-p^2) \\
		&=& 8\frac{r}{\delta m}(q+p)(q-p) \\
		&\leq& 8\frac{r}{\delta m}\cdot 2\cdot q\cdot (2+\delta)m \\
		&\leq& 48\frac{r}{\delta}(s_2+a_2) \\
		&\leq& \frac{96}{\delta}C_{Opt}
\end{array}$$
and
$$\begin{array}{rcl}
  C_{Alg} &=& Da_1+ra_2 \\
	  &\leq& D(p+s_1+q)+ra_2 \\
	  &\leq& D(s_1+2(s_2+a_2))+ra_2 \\
		&\leq& 5C_{Opt}.
\end{array}$$

\textbf{2}. $p>\delta\frac{Dm}{4r}$ and $q\leq\delta\frac{Dm}{4r}$:
If $p<\delta m$, we have $a_2\leq s_2$ as before. The potential difference can be estimated by
$$\begin{array}{rcl}
  \Delta\phi &=& 2Dq-8\frac{r}{\delta m}p^2 \\
	  &\leq& 2D(a_2+s_2)-2Dp \\
		&\leq& 4C_{Opt}-2Dp
\end{array}$$
and
$$\begin{array}{rcl}
  C_{Alg} &=& Da_1+ra_2 \\
	  &\leq& D(p+s_1+q)+ra_2 \\
		&\leq& C_{Opt} + 2Dp.
\end{array}$$

Else, we have
$$\begin{array}{rcl}
  \Delta\phi &=& 2Dq-8\frac{r}{\delta m}p^2 \\
	  &\leq& 2Dq-8rp \\
		&\leq& -6rp
\end{array}$$
and
$$\begin{array}{rcl}
  C_{Alg} &\leq& D(p+s_1+q)+r(q+s_2) \\
	  &\leq& C_{Opt} + 3rp.
\end{array}$$

\newpage

In all of the following cases, we have $p>\delta\frac{Dm}{4r}$ and $q>\delta\frac{Dm}{4r}$.

\textbf{3}. $q-h\leq-(1+\frac{\delta}{2})m$:
Since $s_1\leq m$, we have $q-p\leq-\frac{\delta}{2}m$.
This can be seen by
$$\begin{array}{rcl}
  (1+\frac{\delta}{2})m\leq h-q 
	  \leq p+s_1-q.
\end{array}$$
Therefore we get
$$\begin{array}{rcl}
  \Delta\phi &\leq& 8\frac{r}{\delta m}(q+p)(q-p) \\
	  &\leq& 8\frac{r}{\delta m}(q+p)(-\frac{\delta}{2}m) \\
		&\leq& -4r(q+p)
\end{array}$$
complemented by
$$\begin{array}{rcl}
  C_{Alg} &\leq& C_{Opt} + 2r(p+q).
\end{array}$$

\textbf{4}. $q-h>-(1+\frac{\delta}{2})m$ and $p\geq4m$:
We estimate $a_2$ in two different ways:
First assume our given space is the Euclidean plane.
Since $h-q<(1+\frac{\delta}{2})m \leq \frac{1+\frac{\delta}{2}}{1+\delta}a_1$ if $a_1=(1+\delta)m$,
we get $\frac{\sqrt{\delta}}{2}a_2\leq s_2$ according to Lemma~\ref{leGeo}.
If the given space is the infinite line, $a_1=(1+\delta)m$ directly implies $a_2\leq s_2$ since $P_{Opt}'$
cannot be located between $P_{Alg}'$ and $c$.
For both cases, if $a_1<(1+\delta)m$, then $a_2=0\leq s_2$.

Now, if $\frac{1}{2}p\leq q$ we may estimate
$$\begin{array}{rcl}
  \Delta\phi &\leq& 8\frac{r}{\delta m}(q+p)(q-p) \\
	  &\leq& 8\frac{r}{\delta m}(3q)(3m) \\
		&\leq& 72\frac{r}{\delta}(s_2+a_2) \\
		&\leq& \frac{216}{\delta^{\nicefrac{3}{2}}}C_{Opt}
\end{array}$$
and
$$\begin{array}{rcl}
  C_{Alg} &=& Da_1+ra_2 \\
	  &\leq& D(s_1+3(s_2+a_2))+ra_2 \\
		&\leq& \frac{10}{\sqrt{\delta}}C_{Opt}.
\end{array}$$
Both estimations can be reduced by a factor of $\frac{1}{\sqrt{\delta}}$ when dealing with the infinite line.\\
In case $\frac{1}{2}p> q$ we use
$$\begin{array}{rcl}
  \Delta\phi &\leq& 8\frac{r}{\delta m}(q+p)(q-p) \\
	  &\leq& 8\frac{r}{\delta m}(q+p)(-\frac{1}{2}p) \\
		&\leq& -24\frac{r}{\delta}p
\end{array}$$
and
$$\begin{array}{rcl}
  C_{Alg} &=& Da_1+ra_2 \\
		&\leq& \frac{5}{\sqrt{\delta}}C_{Opt}+Dp.
\end{array}$$

\textbf{5}. $q-h>-(1+\frac{\delta}{2})m$ and $p<4m$:
We get $\frac{\sqrt{\delta}}{2}a_2\leq s_2$ and $a_2\leq s_2$ for the Euclidean plane and the infinite line respectively
as before.
Since $q\leq s_1+p+a_1\leq p+3m$ we have
$$\begin{array}{rcl}
  \Delta\phi &\leq& 8\frac{r}{\delta m}(q+p)(q-p) \\
		&\leq& 8\frac{r}{\delta m}\cdot 11m\cdot q - 8\frac{r}{\delta m}(q+p)\delta\frac{Dm}{4r} \\
		&\leq& 88\frac{1}{\delta}(s_2+a_2)-2D(q+p) \\
		&\leq& \frac{264}{\delta^{\nicefrac{3}{2}}}C_{Opt}-2D(q+p)
\end{array}$$
and
$$\begin{array}{rcl}
  C_{Alg} &=& Da_1+ra_2 \\
		&\leq& \frac{5}{\sqrt{\delta}}C_{Opt}+Dp.
\end{array}$$
Again, both estimations can be reduced by a factor of $\frac{1}{\sqrt{\delta}}$ for the 1-dimensional space.

\noindent In all cases we have $C_{Alg}+\Delta\phi\leq\calO(\frac{1}{\delta^{\nicefrac{3}{2}}})\cdot C_{Opt}$ for the 2-dimensional
and $C_{Alg}+\Delta\phi\leq\calO(\frac{1}{\delta})\cdot C_{Opt}$ for the 1-dimensional Euclidean space.

%%%%%%%%%%%%%%%%%%%%%%%%%%%%%%%%%%%%%%%%%%%%%%%%%%%%%%%%%%%%%%%%%%%%%%%%%%%%%%%%%%
%%%%%%%%%%%%%%%%%%%%%%%%%%%%%%%%%%%%%%%%%%%%%%%%%%%%%%%%%%%%%%%%%%%%%%%%%%%%%%%%%%

\subsection{\texorpdfstring{Analysis for $r\leq D$}{Analysis for r<=D}}
\label{sec:ana2}

For $r\leq D$ we first give a detailed analysis for the Euclidean plane and then
briefly describe how to modify the necessary parts to work for the line segment such that the competitive
ratio becomes $\mathcal{O}(\frac{1}{\delta})$.

We increase the potential $\phi$ from before by a factor of 2, such that
$$\phi(P_{Opt},P_{Alg}):=\left\{\begin{array}{l}
16\frac{r}{\delta m}\cdot d(P_{Opt},P_{Alg})^2,\\
  \tab\text{for }d(P_{Opt},P_{Alg})>\delta\frac{Dm}{4r}\\
4D\cdot d(P_{Opt},P_{Alg}),\\
\tab\text{otherwise}
\end{array}\right..$$

\textbf{1}. Let $p\leq\delta\frac{Dm}{4r}$ and $q\leq\delta\frac{Dm}{4r}$.
First we consider the case $s_2\leq\frac{\sqrt{\delta}}{1+\frac{\delta}{2}}a_2$.
Using Lemma~\ref{leGeo}, we bound the potential difference by
$$\begin{array}{rcl}
  \Delta\phi &=& 4D(q-p) \\
	  &\leq& 4D(q-h+h-p) \\
		&\leq& -4D\frac{1+\frac{\delta}{2}}{1+\delta}a_1 + 4Ds_1.
\end{array}$$

If $a_1=\frac{r}{D}(a_1+a_2)$ then $\Delta\phi\leq -3r(a_1+a_2)+ 4Ds_1$.
We have $C_{Alg}\leq D\cdot\frac{r}{D}(a_1+a_2)+ra_2\leq 2r(a_1+a_2)$.

Otherwise, $a_1=(1+\delta)m$ which gives us $\Delta\phi\leq-4D(1+\frac{\delta}{2})m+ 4Ds_1$.
In this case $q\leq\delta\frac{Dm}{4r}$ can be used to get
$$\begin{array}{rcl}
  C_{Alg} &\leq& D(1+\delta)m + ra_2 \\
	  &\leq& D(1+\delta)m + r(\delta\frac{Dm}{4r}+s_2) \\
		&\leq& 2D(1+\delta)m + C_{Opt}.
\end{array}$$

The second case is $s_2>\frac{\sqrt{\delta}}{1+\frac{\delta}{2}}a_2$.
If $\frac{1}{2}p\leq q$ then
$$\begin{array}{rcl}
  \Delta\phi &\leq& 4D(s_1+p+a_1-p) \\
	  &\leq& 4Ds_1 + 4r(a_1+a_2) \\
		&\leq& 4Ds_1 + 4r(p+s_1+q+a_2) \\
		&\leq& 8Ds_1 + 4r(3s_2+4a_2) \\
		&\leq& \frac{44}{\sqrt{\delta}} C_{Opt}
\end{array}$$
and
$$\begin{array}{rcl}
  C_{Alg} &\leq& r(p+s_1+q) + 2ra_2 \\
	  &\leq& \frac{10}{\sqrt{\delta}}C_{Opt}.
\end{array}$$
Else $\frac{1}{2}p> q$ and we have
$$\begin{array}{rcl}
  \Delta\phi &\leq& -2Dp
\end{array}$$
and
$$\begin{array}{rcl}
  C_{Alg} &\leq& D(p+s_1+q) + ra_2 \\
	  &\leq& \frac{2}{\sqrt{\delta}}C_{Opt} + \frac{3}{2}Dp.
\end{array}$$

We use the same argumentation for $p>\delta\frac{Dm}{4r}$
since in this case $-16\frac{r}{\delta m}p^2\leq - 4Dp$.

\textbf{2}. $p>\delta\frac{Dm}{4r}$ and $q>\delta\frac{Dm}{4r}$:
As before we start with the case that $s_2\leq\frac{\sqrt{\delta}}{1+\frac{\delta}{2}}a_2$.
We have
$$\begin{array}{rcl}
  \Delta\phi &=& 16\frac{r}{\delta m}(q^2-p^2) \\
	  &\leq& 16\frac{r}{\delta m}(q+p)(q-p) \\
		&\leq& 16\frac{r}{\delta m}(q+p)(-\frac{1+\frac{\delta}{2}}{1+\delta}a_1+s_1).
\end{array}$$
If $a_1=(1+\delta)m$ then
$\Delta\phi\leq -8r(q+p)$ and
$C_{Alg}\leq2r(a_1+a_2)\leq \frac{4}{\sqrt{\delta}}C_{Opt}+2r(p+q)$.

Else, $a_1=\frac{r}{D}(a_1+a_2)<(1+\delta)m$
and using $a_1+a_2\leq2\frac{Dm}{r}$ we get
$$\begin{array}{rcl}
  \Delta\phi &\leq& 16\frac{r}{\delta m}(q+p)s_1 - 16\frac{r}{\delta m}(q+p)\frac{1+\frac{\delta}{2}}{1+\delta}a_1 \\
	  &\leq& 16\frac{r}{\delta m}s_1(s_1+a_1+2s_2+2a_2) - 16\frac{r}{\delta m}\cdot \delta\frac{Dm}{4r}\cdot\frac{3}{4}\frac{r}{D}(a_1+a_2) \\
		&\leq& 64\frac{D}{\delta}s_1+32\frac{r}{\delta}s_2 -3r(a_1+a_2) \\
		&\leq& \frac{64}{\delta}C_{Opt} - 3r(a_1+a_2)
\end{array}$$
For the online algorithm we have
$C_{Alg}\leq2r(a_1+a_2)$.

Now consider the case $s_2>\frac{\sqrt{\delta}}{1+\frac{\delta}{2}}a_2$.
If $\frac{1}{2}p\leq q$ we use it to get
$$\begin{array}{rcl}
  \Delta\phi &=& 16\frac{r}{\delta m}(q+p)(q-p) \\
		&\leq& 16\frac{r}{\delta m}(3q)(3m) \\
		&=& 432\frac{r}{\delta^{\nicefrac{3}{2}}}s_2
\end{array}$$
and
$$\begin{array}{rcl}
  C_{Alg} &\leq& 2r(a_1+a_2) \\
	  &\leq& \frac{5}{\sqrt{\delta}}C_{Opt}.
\end{array}$$

Otherwise we have
$$\begin{array}{rcl}
  \Delta\phi &=& 16\frac{r}{\delta m}(q+p)(q-p) \\
		&\leq& 16\frac{r}{\delta m}(q+p)(-q) \\
		&=& -4D(q+p)
\end{array}$$
and
$C_{Alg}\leq \frac{2}{\sqrt{\delta}}C_{Opt}+D(p+q)$.

Again, the arguments also apply to $p\leq\delta\frac{Dm}{4r}$ due to
$-4Dp\leq-16\frac{r}{\delta m}p^2$.

To get the bound for the line segment, in each of the two big cases replace the distinction whether $s_2\leq\frac{\sqrt{\delta}}{1+\frac{\delta}{2}}a_2$
by $s_2\leq a_2$.
Also the estimation of $q-h$ under the use of Lemma~\ref{leGeo} may be replaced by $q-h\leq -a_1$.

Again, in all cases we have $C_{Alg}+\Delta\phi\leq\calO(\frac{1}{\delta^{\nicefrac{3}{2}}})\cdot C_{Opt}$ for the 2-dimensional
and $C_{Alg}+\Delta\phi\leq\calO(\frac{1}{\delta})\cdot C_{Opt}$ for the 1-dimensional case.

%%%%%%%%%%%%%%%%%%%%%%%%%%%%%%%%%%%%%%%%%%%%%%%%%%%%%%%%%%%%%%%%%%%%%%%%%%%%%%%%%%
%%%%%%%%%%%%%%%%%%%%%%%%%%%%%%%%%%%%%%%%%%%%%%%%%%%%%%%%%%%%%%%%%%%%%%%%%%%%%%%%%%

\subsection{General Case and Answer-First Variant}
\label{sec:ext}
In this section we briefly describe how to modify the analysis such that the result
for general $R_{min}$ and $R_{max}$ follows.
Furthermore we show that the Move-To-Center Algorithm is also almost optimal
in the Answer-First Variant.

The general result can be derived as follows:
We replace the fixed number of requests $r$ in the potential function by the maximum number of
requests $R_{max}$.
In the cases where the potential is used to cancel out the costs of the algorithm this is then still
possible.
However if the potential difference is positive it may add a term which is $\mathcal{O}(\frac{R_{max}}{R_{min}})$ times the optimal costs.

\begin{theorem}
Let $r\geq D$ be the fixed number of requests per time step.
Algorithm \textit{MtC} is $\mathcal{O}(\frac{1}{\delta^{\nicefrac{3}{2}}}\cdot\frac{r}{D})$-competitive utilizing a maximum moving distance of $(1+\delta)m$
in the Answer-First Variant.
\end{theorem}

\begin{proof}
We relate the costs of the algorithm in the Answer-First Variant to the costs for the same request sequence in the original model.
The optimal solution may be assumed to be the same for both model variants given the same request sequence, since we can insert $r$ dummy requests
on $P_0$ at the beginning of the sequence which allows the optimal solution to operate the same as in the Answer-First model and only has additional costs of at most $rm$.

In case $r\leq D$,
the additional term $ra_1$ in the costs of the algorithm 
can be estimated via $ra_1\leq Da_1$ which implies an increase of the algorithms cost by at most a factor 2.
For the case $r>D$,
the additional term in the costs of the algorithm is $ra_1$, which implies the costs of the algorithm increase by a factor of at most $2\frac{r}{D}$.
\end{proof}

%%%%%%%%%%%%%%%%%%%%%%%%%%%%%%%%%%%%%%%%%%%%%%%%%%%%%%%%%%%%%%%%%%%%%%%%%%%%%%%%%%
%%%%%%%%%%%%%%%%%%%%%%%%%%%%%%%%%%%%%%%%%%%%%%%%%%%%%%%%%%%%%%%%%%%%%%%%%%%%%%%%%%

\section{The Moving Client Variant}

We propose a variant in which the requests are posed by agents located in the network (or physically in the Euclidean plane) which also move at a limited speed.
As an example, consider a disaster scenario where the helpers form an ad-hoc network with their smartphones in order to coordinate their efforts.
In these scenarios, data is sometimes physically transported or a mobile signal station is used.
In a sense, the server in this case is supposed to follow the agents around in order to make best use of the resource.
We focus on only having one agent in the network, however our results can be modified to also work for multiple agents by similar arguments as in the original problem.

The complete model is as follows:
A server as well as an agent are initially located at some point $P_0$ in the Euclidean plane.
In step $t$, the server can be moved from $P_{t-1}$ to $P_{t}$ with $d(P_{t-1},P_t)\leq m_s$.
The agent starts at point $A_0=P_0$ and moves from $A_{t-1}$ to $A_{t}$ with $d(A_{t-1},A_t)\leq m_a$ in round $t$.
The costs of a sequence of length $T$ is $\sum_{t=1}^{T}\left(D\cdot d(P_{t-1},P_t)+d(P_t,A_t)\right)$.
In the online variant, the position $A_t$ gets revealed at the beginning of round $t$, i.e. the agent moves before the server.

We first show that the competitive ratio is unbounded if $m_a$ is larger than $m_s$.

\begin{theorem}
Let $\varepsilon>0$ and $m_a=(1+\varepsilon)m_s$. No randomized online algorithm can be better than $\Omega(\sqrt{T}\frac{\varepsilon}{1+\varepsilon})$-competitive
in the Moving Client Variant against an oblivious adversary.
\end{theorem}

\begin{proof}
The construction is similar to the bound in Theorem~\ref{ThInfinite}.
We divide the input sequence into two major phases.
In the first phase, for some value $x$ chosen later, the adversary moves the server $x\cdot\frac{m_a}{m_s}$ rounds by a distance of $m_s$
in one of two opposite directions determined by a fair coin.
The requesting agent stays at $P_0$ and only moves to the position of the adversary for the last $x$ time steps.
The costs for the adversary in this step are at most $Dxm_a+x^2\frac{m_a^2}{m_s}$.
By the end of this phase, with probability $\frac{1}{2}$ the online algorithm has a distance of at least $x(m_a-m_s)=x\varepsilon  m_s$ to the agent and the adversary,
since it can only move $x$ steps towards the agent after the outcome of the random experiment is revealed.

In the second phase the adversary and the agent continue to move in the same direction by a distance of $m_s$ in each round.
Hence the costs of the online algorithm are at least $(T-x\frac{m_a}{m_s})(x\varepsilon  m_s)$ while
the costs of the adversary are at most $D(T-x\frac{m_a}{m_s})m_s$.
By setting $x=\sqrt{T}\cdot\frac{m_s}{m_a}$ we get a competitive ratio of $\Omega(\sqrt{T}\frac{\varepsilon}{1+\varepsilon})$.
\end{proof}

From Theorem~\ref{ThMain} it follows directly that by using an augmented maximum moving distance of $(1+\delta)m_s$ for the online algorithm,
we can achieve a competitive ratio independent of $T$.

\begin{corollary}
MtC is
$\calO(\frac{1}{\delta^{\nicefrac{3}{2}}})$-competitive in the Moving Client Variant using
an augmented maximum moving distance of $(1+\delta)m_s$ for some $\delta\in(0,1]$.
\end{corollary}

Until now, we only discussed the case $m_a>m_s$.
Note that in case $m_a\leq m_s$ standard solutions to the Page Migration Problem still do not apply, since they require
moving to a specific point after collecting a batch of requests.
This point may still lie outside the allowed moving distance $m_s$.
However we can show, that the MtC-Algorithm used to solve our original problem achieves a constant competitive ratio.
We stick to the case $m_s=m_a$ for easier readability; however it is easy to see that the result applies to the more general
case $m_s\geq m_a$ as well.

The algorithm in our case does the following:
Upon receiving the position of the agent at point $A_t$,
move the server a distance of $\min(m_s,\frac{1}{D}d(P_{t-1},A_t))$ towards $A_t$. 

\begin{theorem}
The MtC Algorithm is $\mathcal{O}(1)$-competitive in the Moving Client Variant for $m:=m_s=m_c$.
\end{theorem}

\begin{proof}
We define the potential $\phi(P_{Alg},P_{Opt}):=2^{\frac{3}{2}}D\cdot d(P_{Alg},P_{Opt})$.
We use the same notation as in the proof of Theorem~\ref{ThMain}, shown in Figure~\ref{fig:geo1}.
Hence
$$\Delta\phi=2^{\frac{3}{2}}D\cdot(q-p)\leq 2^{\frac{3}{2}}D\cdot(q-h+s_1).$$
We also use a simpler version of Lemma~\ref{leGeo} to show that $h-q\geq \frac{1}{\sqrt{2}}a_1$ if $s_2\leq a_2$, which can be seen
by setting $\epsilon=1$ in the proof of the lemma.
We distinguish two cases:

If $s_2\leq a_2$, then
$\Delta\phi \leq 2^{\frac{3}{2}}Ds_1 - 2Da_1$.
Either $a_1=\frac{1}{D}(a_1+a_2)$ and hence
$$C_{Alg}+\Delta\phi= D\cdot\frac{1}{D}(a_1+a_2)+a_2+\Delta\phi\leq 2^{\frac{3}{2}}Ds_1$$
or $a_1=m$ and
$$C_{Alg}+\Delta\phi\leq 2Dm+\Delta\phi\leq 2^{\frac{3}{2}}Ds_1$$
since the algorithm maintains a distance of at most $Dm$ to the request.

For the second case, which is $a_2<s_2$, we make use of
$C_{Alg}\leq p+s_1+q+2s_2$.
If $q\leq(1-\frac{1}{2^{\nicefrac{3}{2}}})$ then $\Delta\phi\leq-Dp$ and
$$C_{Alg}+\Delta\phi\leq s_1+4s_2.$$
Else
$$\begin{array}{rcl}
  \Delta\phi &\leq& 2^{\frac{3}{2}}D(s_1+p+a_1-p) \\
	  &\leq& 2^{\frac{3}{2}}Ds_1 + 2^{\frac{3}{2}}(a_1+a_2) \\
		&\leq& 2^{\frac{3}{2}}Ds_1 + 2^{\frac{3}{2}}(p+s_1+q+a_2) \\
		&\leq& 2^{\frac{5}{2}}Ds_1 + 2^{\frac{3}{2}}(3s_2+4a_2) \\
		&\leq& 36 C_{Opt}
\end{array}$$
and
$$C_{Alg}\leq s_1 + 3q+2s_2\leq Ds_1+8s_2.$$

\end{proof}

%%%%%%%%%%%%%%%%%%%%%%%%%%%%%%%%%%%%%%%%%%%%%%%%%%%%%%%%%%%%%%%%%%%%%%%%%%%%%%%%%%
%%%%%%%%%%%%%%%%%%%%%%%%%%%%%%%%%%%%%%%%%%%%%%%%%%%%%%%%%%%%%%%%%%%%%%%%%%%%%%%%%%

\section{Conclusion}
We have seen that a simple, deterministic algorithm is sufficient to get an almost optimal competitive ratio in various versions of our model.
For the Euclidean space of dimension 1, we have an asymptotically optimal competitive ratio which can not be beaten even by a randomized
algorithm against an oblivious adversary.

For the Euclidean space of higher dimensions, we miss the optimal competitive ratio only by a factor of $\frac{1}{\sqrt{\delta}}$.
We conjecture that this remaining gap between the upper and the lower bound can be closed towards the lower bound, but it remains
an open problem to design an online algorithm and / or provide an analysis to achieve the better competitive ratio.

It seems an interesting question if the idea of limiting the movement of resources within a time slot
also can be applied to other popular models such as the $k$-Server Problem (effectively turning it into
the Page Migration Problem with multiple pages).

It may also be possible to extend existing online problems without a concept of movement by introducing
a limited movement like in our model. In problems like the Online Facility Location Problem, this might
give possibilities to the online algorithms to slightly improve upon decisions where to open a facility.

Furthermore the concept of allowing only a limited configuration change might be applicable to any problem
which belongs to the class of Metrical Task Systems where in every step configurations may be changed to
lower costs for answering a certain type of requests which have to be served by the system.

%%%%%%%%%%%%%%%%%%%%%%%%%%%%%%%%%%%%%%%%%%%%%%%%%%%%%%%%%%%%%%%%%%%%%%%%%%%%%%%%%%
%%%%%%%%%%%%%%%%%%%%%%%%%%%%%%%%%%%%%%%%%%%%%%%%%%%%%%%%%%%%%%%%%%%%%%%%%%%%%%%%%%

\newpage
\bibliographystyle{plain}
\bibliography{references}

\end{document}